\documentclass[onecolumn]{IEEEtran}

\usepackage{amsfonts}
\usepackage{amsmath}
\usepackage{amssymb}
\usepackage{graphicx}
\usepackage{float}
\usepackage{booktabs}
\usepackage[mathscr]{eucal}
\usepackage{circuitikz}
\usepackage{amsfonts}
\usepackage{float}
\usepackage[colorlinks,linkcolor=blue]{hyperref}
\usepackage[latin1]{inputenc}
\usepackage{amsmath}
\usepackage{graphicx}
\usepackage{amssymb}
\usepackage{amsthm}
\usepackage{verbatim}
\usepackage[labelfont=bf]{caption}
\usepackage{enumerate}
\usepackage{graphicx,subfigure}
\newtheorem{lemma}{Lemma}
\newtheorem{proposition}{Proposition}
\newtheorem{corollary}{Corollary}
\newtheorem{fact}{Fact}

\newtheorem{remark}{Remark}

\newtheorem{assumption}{Assumption}

\def\begcen{\begin{center}}
	\def\endcen{\end{center}}

\newcommand{\col}{ \mbox{col} }

\def\cale{{\cal E}}

\def\calr{{\cal R}}
\def\calj{{\cal J}}

\def\cale{{\cal E}}

\def\call{{\cal L}}

\def\hal{{1 \over 2}}

\def\liminf{\lim_{t \to \infty}}

\def\L2{{\cal L}_2}
\def\L2e{{\cal L}_{2e}}

\def\rea{\mathbb{R}}

\def\sign{\mbox{sign}}

\def\x{{x}}

\def\begmat#1{\begin{bmatrix}#1\end{bmatrix}}
\def\begali#1{\begin{align}{#1}\end{align}}
\def\begalis#1{\begin{align*}{#1}\end{align*}}

\def\begequarr{\begin{eqnarray}}
	\def\endequarr{\end{eqnarray}}
\def\begequarrs{\begin{eqnarray*}}
	\def\endequarrs{\end{eqnarray*}}
\def\begarr{\begin{array}}
	\def\endarr{\end{array}}
\def\begequ{\begin{equation}}
	\def\endequ{\end{equation}}
\def\lab{\label}
\def\begdes{\begin{description}}
	\def\enddes{\end{description}}
\def\begenu{\begin{enumerate}}
	\def\begite{\begin{itemize}}
		\def\endite{\end{itemize}}
	\def\endenu{\end{enumerate}}

\def\lef[{\left[\begin{array}}
	\def\rig]{\end{array}\right]}

\def\begcen{\begin{center}}
	\def\endcen{\end{center}}
\def\begrem{\begin{remark}\rm}
	\def\endrem{\end{remark}}
\def\begassum{\begin{assumption}}
	\def\endassum{\end{assumption}}
\def\begassums{\begin{assumption*}}
	\def\endassums{\end{assumption*}}
\def\begassu{\begin{ass}}
	\def\endassu{\end{ass}}
\def\beglem{\begin{lemma}}
	\def\endlem{\end{lemma}}
\def\begcor{\begin{corollary}}
	\def\endcor{\end{corollary}}
\def\begfac{\begin{fact}}
	\def\endfac{\end{fact}}

\DeclareMathOperator*{\argmin}{arg\,min}

\def\liminf{\lim_{t \to \infty}}

\def\L2e{{\cal L}_{2e}}

\def\rea{\mathbb{R}}

\def\sign{\mbox{sign}}

\def\col{\mbox{col}}
\def\hal{{1 \over 2}}

\def\begsubequ{\begin{subequations}}
	\def\endsubequ{\end{subequations}}

\newcommand\dif{\mathrel{\overset{\makebox[0pt]{\mbox{\normalfont\tiny ${d \over dt}$}}}{\Longrightarrow}}}
\newcommand\ttt{\mathrel{\overset{\makebox[0pt]{\mbox{\normalfont\tiny $\times \sign(y_2)$}}}{\Longrightarrow}}}
\newcommand\fff{\mathrel{\overset{\makebox[0pt]{\mbox{\normalfont\tiny $F(p)(\cdot)$}}}{\Longrightarrow}}}
\newcommand\hhh{\mathrel{\overset{\makebox[0pt]{\mbox{\normalfont\tiny $\times h_3$}}}{\Longrightarrow}}}

\usepackage{xcolor}

\begin{document}
	\title{Output Voltage Regulation of a Fuel Cell/Boost Converter System: A PI-PBC Approach}
	\author{
		Rafael Cisneros, Romeo Ortega, Carlo A. Beltrán, Diego Langarica-Córdoba, Luis. H. Díaz-Saldierna
		\thanks{R. Cisneros and R. Ortega are with 
			Departamento Académico de Sistemas Digitales, Instituto Tecnológico Autónomo de México (ITAM), 01080 Ciudad de México, Mexico. Emails: \{rcisneros,romeo.ortega\}@itam.mx.}
		\thanks{C.A. Beltrán and D. Langarica-Córdoba are with  Faculty of Sciences, Autonomous University of San Luis Potosi (UASLP), 78295 San Luis Potosi, Mexico. E-mails: \{carlo.beltran,diego.langarica\}@uaslp.mx. }
		\thanks{L. H. Díaz-Saldierna is with Institute for Scientific and Technological Research of San Luis Potosi (IPICYT), 78216 San Luis Potosi, Mexico. E-mail: ldiaz@ipicyt.mx.}
	}
	
	\maketitle

\begin{abstract}In this paper we consider the problem of voltage regulation of a proton exchange membrane fuel cell connected to an uncertain load through a boost converter. We show that, in spite of the inherent  nonlinearities in the current-voltage behaviour of the fuel cell, the voltage of a fuel cell/boost converter system can be regulated with a simple {\em proportional-integral} (PI) action designed following the {\em Passivity-based Control} (PBC) approach \cite{HERetal,ORTpidbook}. The system under consideration consists of a DC-DC converter interfacing a fuel cell with a resistive load. We show that the output voltage of the converter converges to its desired constant value for all the systems initial conditions---with convergence ensured {\em for all} positive values of the PI gains. This latter feature facilitates the, usually difficult, task of tuning the gains of the PI. An Immersion and Invariance parameter estimator  \cite{ASTKARORTbook} is afterwards proposed which allows the operation of the PI-PBC when the {\em load is unknown}, maintaining the output voltage at the desired level.  The stable operation of the overall system is proved and the approach is validated with extensive numerical simulations considering real-life scenarios, where a robust behavior in spite of load variations is obtained.
\end{abstract}


%
\section{Introduction}
\lab{sec1}
The proton exchange membrane fuel cell (PEMFC) plus DC/DC boost converter system plays a crucial role in the operation of many practical applications, including electric vehicle drive-train systems \cite{NAMbook} and DC microgrids \cite{SCHetal}. This configuration interconnects the PEMFC stack with the DC link voltage of the traction inverter or the microgrid. The PEMFC DC/DC converter enables the energy flow between these two electrical subnets over a wide voltage range, hence a critical task to be accomplished is to design a control strategy that ensures an effective {\em voltage regulation}. This task is complicated by several factors.
\begite
\item As is well-known \cite{KASSCHVERbook},  the dynamics of the boost converter is described by a bilinear system. Moreover, the zero dynamics with respect to the voltage output of this system is {\em unstable},  that is, the system is non-minimum phase.  

\item The relation between the current and the voltage of the PEMFC---the so-called polarization curve---is highly {\em nonlinear} and uncertain \cite{LARDICbook}.

\item The load that is connected to the converter, usually modeled by a simple resistance, is effectively time-varying and {\em uncertain}. Moreover, the parasitic resistance of the boost inductor is also poorly known and time-varying.
\endite 

In this paper we show that, in spite of all these complicated features, it is possible to regulate effectively the output voltage of the boost converter with a {\em simple PI-PBC}, that can be made adaptive incorporating an Immersion and Invariance (I\&I) parameter estimator  \cite{ASTKARORTbook}, in order to estimate the inductor parasitic resistance and the load resistance, and consequently the corresponding equilibrium point. The PI-PBC used here is based on the PI-PBC first reported in \cite{HERetal} and has been widely accepted by the power electronics and automatic control communities, being cited more than 160 times in Google Scholar---see \cite{ORTpidbook,ZONetalijrnlc22,WeiHe2022}  for a recent list of references to this work. A fundamental property of the PI-PBC is that convergence is ensured {\em for all} positive values of the PI gains, facilitating the, usually difficult, task of tuning the gains of the PI. To take care of the fact that the boost converter is fed by a PEMFC, and not a constant DC source as it is in  \cite{HERetal}, it is necessary to modify the design incorporating the PEMFC polarization curve characteristic. It is shown in the paper that it is possible to exploit the monotonic nature of this curve to prove the stability of the new PI-PBC design even including the I\&I parameter estimator.

As stated before, due to the non-minimum phase behavior of the boost converter, direct voltage-mode control (VMC) is not suitable, instead current-mode control (CMC) is preferred to achieve voltage regulation indirectly, through inductor current regulation. For instance, a CMC strategy for the PEMFC plus boost converter system is detailed in \cite{Diaz2017}, where  two {\em linear control} loops are implemented: an outer PI voltage control loop for current reference generation and a current loop for duty cycle generation, consisting of a high gain compensator and a low-pass filter. This last technique is extended in \cite{Diaz2021}, when a quadratic boost converter is connected to a PEMFC instead of boost converter in order to achieve higher output voltages. Furthermore, in \cite{Sira2012} a linear state-feedback controller based on integral passive output feedback is designed for the PEMFC plus boost converter system, which semi-globally stabilizes the output voltage regulation error to zero. On the other hand, {\em nonlinear controllers} are reported in the literature for the system under study. For example,  an adaptive backstepping control for a fuel cell/boost converter system is designed in \cite{Yuz2018}. The overall scheme consist of two control loops: an inner current loop based on dynamic backstepping approach plus an I\&I estimator for load approximation and an outer voltage loop based on a PI action over the output voltage. Additionally, a CMC combined with PBC scheme for output voltage regulation in a fuel cell-boost converter system is proposed in \cite{CarloBeltran2023}. In this case, the current loop is designed with PBC to take advantage of the passive map of the control signal to the inductor current, meanwhile the voltage loop is designed to generate the desired current reference with a PI action over the voltage error.  In both previous cases \cite{Yuz2018,CarloBeltran2023}, due to the merging of the two control loops through the time-derivative of the current reference, these approaches result in an intricate control laws where the proportional gain of the voltage loop is bounded by system parameters in order to assure a well-defined control law. Moreover, the aforementioned linear and non-linear approaches use two simplified fuel cell polarization curve models to represent the PEMFC static current-voltage characteristics \cite{CarloBeltran2023}. In particular the rational model proposed in \cite{Shahin2010} and the two-termed power function proposed in \cite{Yuz2018IECON}, which are obtained from experimental data and requires the knowledge of three parameters to be implemented. 

Different from the control design approaches above, in this paper, we propose a voltage regulation scheme
based on a simple PI-PBC scheme for the fuel cell-boost converter system.  We show that the output voltage of the converter converges to its desired constant value for all initial conditions---with convergence ensured {\em for all} positive values of the PI gains. This feature facilitates the, usually difficult, task of tuning the gains of the PI. Additionally, an Immersion and Invariance parameter estimator  \cite{ASTKARORTbook} is afterwards proposed which allows the operation of the PI-PBC when the {\em inductor parasitic resistance and load are unknown}, maintaining the output voltage at the desired level. Furthermore, the well-known Larminie-Dicks model in \cite{LARDICbook} is used to represent the PEMFC static current-voltage characteristics. This model takes into account four different internal  FC voltage phenomena, i.e., the open-circuit voltage, the activation voltage loss, the ohmic voltage loss and the voltage concentrations loss. The following are the main contributions of this research regarding to the output voltage regulation of the PEMFC-boost converter system: 

\begin{itemize}
	\item The design of a simple and easy-to-tune PI-PBC scheme for output voltage regulation of a fuel cell system.
	\item Different from other approaches, the proposed scheme implements only one control loop, which simplifies its practical implementation.
	\item The controller is well-defined, and the control gains can be selected freely, with overall stability guaranteed. 
	\item The proposed control scheme uses the Larminie-Dicks model to describe the current-voltage polarization curve of the PEMFC. To the best of the authors knowledge, this expression has not been used for control purposes in the literature before. 
\end{itemize}

The remainder of the paper is organized as follows. The description of the PEMFC plus DC/DC boost converter system is given in Section \ref{sec2}. In Section \ref{sec3}, we present the proposed PI-PBC, for the case of known parameters. Afterwards, section \ref{sec4} presents its adaptive version.  Simulation results, which illustrate the closed-loop performance of the proposed adaptive PI-PBC, are presented in Section \ref{sec5}. Finally, the paper is wrapped-up with concluding remarks and future work in  Section \ref{sec6}.

\section{System description}
\lab{sec2}
A schematic diagram of the PEMFC+boost converter system considered  in this work is given in Fig \ref{fig1}.  The PEMFC current and voltage are $i_{fc}$ and $v_{fc}$, respectively. The converter inductor current and output voltage are $i_L$ and $v_o$, respectively. Additionally, the system components are the fuel cell coupling capacitor $C_{fc}$, the inductor $L$ and its parasitic resistance $R_P$, the output capacitor $C$ and the load $R_L$, considered purely resistive. Finally, the voltage boosting action ($v_o > v_{fc}$) is achieved through the duty cycle (control signal) $D$, which is used with a pulse width modulation to generate a trigger signal to drive the switches of the converter. 

In this section we give its mathematical description assuming that the behavior of the PEMFC is described by its polarization curve---an assumption that is reasonable given the big differences in time scales between the PEMFC and the boost converter. It is, furthermore, assumed that the boost load can be described by a simple linear resistor, which is however uncertain and time-varying---hence it is necessary to estimate its value on-line.

\begin{figure*}[t!]
	\centering
	\includegraphics{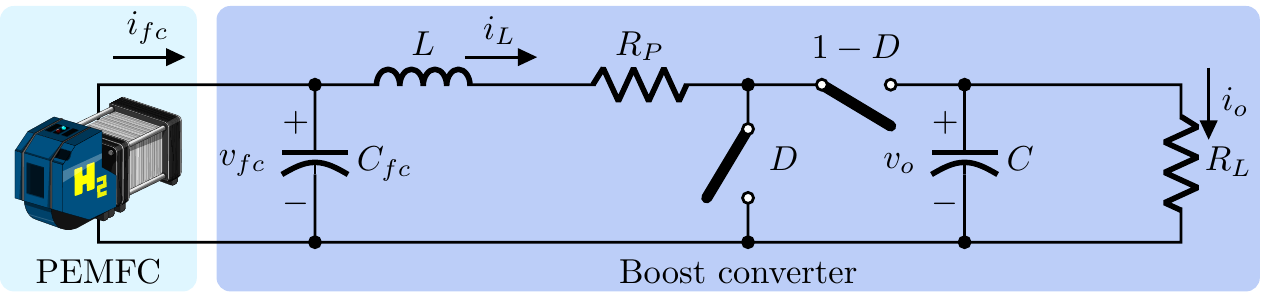}
	\caption{The PEMFC boost converter system}
	\lab{fig1}
\end{figure*}
\begin{figure}[t!]
	\centering
	\begin{circuitikz}[american voltages,scale=0.8]
		\draw
		(0,3) to [V, l_=$c_1$] (0,0)
		to [short,-*] (7,0)
		(0,3) to [R, l_=$c_2$, v^>=$\hspace{0.2cm}$] (2,3)
		to [short, i_=$i_{fc}$] (3,3)
		(6,3) to [cV, l=$c_3\ln i_{fc}+c_5e^{c_4i_{fc}}$] (3,3)
		(6,3) to [short,-*] (7,3)
		(7,0) to [open, v_<=$v_{fc}$] (7,3); 
	\end{circuitikz}
	\caption{A circuit diagram representation of \eqref{Vfc}.}
	\label{pemfc_circ}
\end{figure}
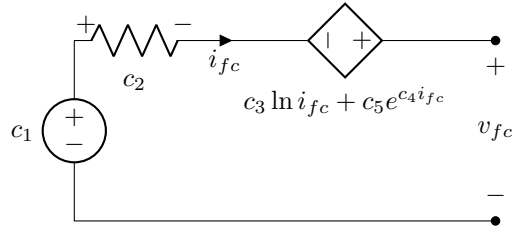

\subsection{The PEMFC}
\lab{subsec21}
The polarization curve of the PEMFC, which describes the relation between its current $i_{fc}$ and its voltage $v_{fc}$, is given by the function \cite[eq. 3.14]{LARDICbook}
\begin{equation}\label{Vfc}
	v_{fc}=V_{fc}(i_{fc}):=c_1 - c_2 \ln (i_{fc}) - c_3i_{fc} - c_5e^{c_4i_{fc}}
\end{equation}
where  $v_{fc}>0,i_{fc}>0$ and all the parameters satisfy  $c_i\geq0,\;i=1,\dots,5$. A circuit representation of this model is given in Fig. \ref{pemfc_circ}. Moreover, a graph of this function is shown in Fig. \ref{vfc_ifc} found in Section \ref{sec4}. 

The lemma below reveals a monotonicity property of the curve, which is essential for the stability proof of the PI-PBC.

\begin{lemma}\em
	\lab{lem1}
	The function $V_{fc}(i_{fc})$ given in \eqref{Vfc} is {\em monotonically decreasing} in the positive real axis. That is, for all $a>0,b>0$, it satisfies
	\begin{align}
		\lab{xi}
		(a - b)[V_{fc}(a)-V_{fc}(b)] \leq 0.
	\end{align}
\end{lemma}
\begin{proof}
	The proof is established proving that   $V_{fc}'(i_{fc}) \leq 0$, which can be computed as
	$$
	V_{fc}'(i_{fc}) =-\frac{c_2}{i_{fc}}- c_3  - c_5c_4e^{c_4i_{fc}}
	$$
	which is clearly non-positive.
\end{proof}
\subsection{The overall system: PEMFC+boost converter and assignable equilibrium set}
\lab{subsec22}
%
The PEMFC+boost converter system is modeled  by the following set of equations
\begin{subequations}\label{sys1}
	\begin{align}
		C_{fc}\dot v_{fc}= &i_{fc} - i_L \\
		L \dot i_{L} = & v_{fc} -R_L i_{L} - (1-D ) v_{o}\\
		C\dot v_o =& -\frac{1}{R_L}v_o + (1-D)i_{L}& 
	\end{align}
\end{subequations}
where all signal are identified in Fig \ref{fig1} and $D \in (0,1)$  is the converter's duty ratio. The equations \eqref{sys1} can be expressed in a compact form as introduced in the following fact, whose proof follows trivially from the equations themselves.

\begin{fact}\em
	\lab{fac1}
	The PEMFC+boost converter system \eqref{sys1} admits  the following representation 
	\begequ
	\lab{sys2}
	\mathcal{Q} \dot x= (\calj_0+\calj_1 u-\calr)x+d(x_1)
	\endequ
	with the definitions
	$$
	x:=\begmat{v_{fc} \\ i_L \\ v_o},\;\mathcal{Q}=\begin{bmatrix}C_{fc}&0&0\\0&L&0\\0&0&C\end{bmatrix},\;\calj_0:=\begmat{0&-1&0\\1& 0 &0 \\ 0&0&0},\;\calj_1:=\begmat{0&0&0\\0&0&-1 \\ 0&1&0},\;\calr:=\begmat{0&0&0\\0&\theta_1&0 \\ 0&0&\theta_2},
	$$
	$$
	\;d(x_1):=\begmat{ d_1(x_1) \\ 0 \\ 0},\;\theta=\begin{bmatrix}R_p\\\frac{1}{R_L}\end{bmatrix},\; u:=1-D,
	$$
	and 
	\begin{equation}\label{ifc}
		d_1(x_1):=V^{I}_{fc}(x_1)=i_{fc},
	\end{equation}
	where $V^{I}_{fc}(x_1)$ denotes the left inverse of the mapping $V_{fc}(x_1)$, that is $V^{I}_{fc}(V_{fc}(x_1))= x_1$.\footnote{This mapping is well defined because of the strict monotonicity of $V_{fc}(x_1)$.} 
\end{fact}

A first step in the design of a controller whose objective is to drive the system towards a desired constant equilibrium is to characterize the set of assignable equilibria which are compatible with the control objective---in this case driving the output voltage $x_3$ towards a desired value, denoted $x_3^\star >0$. This task is carried out in the lemma below.    

\begin{lemma}\em
	\lab{lem2}
	Fix a desired value for the output  voltage $x_3=x_3^\star  >0$. The set of assignable equilibrium points of the system \eqref{sys2} compatible with the control objective is defined as 
	\begin{equation}\label{eq}
		\cale:=\left\{x\in\mathbb{R}^3 \;|\;p(x_1,x_3,\theta)=0,\;x_2=d_1(x_1) ,\;x_3=x_3^\star \right\}
	\end{equation}
	where 
	$$
	p(x_1,x_3,\theta):=d_1(x_1)x_1-\theta_2x_{3}^2-\theta_1 d^2_1(x_1).
	$$
\end{lemma}

\begin{proof}
	To identify this set we first eliminate the control signal from the system dynamics. Towards this end, notice that  \eqref{sys2} may be written in the affine form $\dot x=f(x)+g(x)u$ with
	$$
	f(x):=(\calj_0-\calr)x+d(x_1)=\begmat{-x_2 +d_1(x_1) \\ x_1-\theta_1 x_2 \\ -\theta_2 x_3}, \;g(x):=\calj_1 ux=\begmat{0 \\ -x_3 \\ x_2}.
	$$
	Consequently, to eliminate the control we multiply $f(x)$ by a full-rank, left-annihilator of $g(x)$, namely the  matrix
	$$
	g^\perp(x):=\begmat{1 & 0 & 0\\ x_1 & x_2 & x_3}.
	$$
	The set \eqref{eq} is obtained setting $x_3= x_3^\star $ and evaluating 
	$$
	g^\perp(x) f(x)=0.
	$$
\end{proof}

\begin{remark}
\lab{rem1}
We make the important observation  that the assignable equilibrium points are {\em independent} of the boost converter parameters $C_{fc},L$ and $C$. They depend only on the parameters of the PEMFC polarization curve \eqref{Vfc}---which can be experimentally obtained off-line---and the resistance parameters $\theta$. An element of the set $\cale$, for the case when this parameters  are {\em known} will be denoted  via the constant vector  $x^\star \in  \rea^3_+$. The error between the state of the system and an equilibrium will be denoted as
$$
\tilde x:=x-x^\star.
$$ 
In Section \ref{sec4} we propose to estimate these parameters on-line, for this {\em adaptive} case, we will denote the elements of the set $\cale$ as $\hat x^\star$. The error between the actual equilibrium value $x^\star$ and its estimate will be denoted 
\begequ
\lab{e}
e:=\hat x^\star-x^\star.
\endequ 
\end{remark}
\section{Proposed PI-PBC: Known Parameter Case}
\lab{sec3}
%
The proposition below defines the proposed PI-PBC for the case of known parameters $\theta$. The adaptive version, where these parameters are estimated on-line, is given in Section \ref{sec4}.

\begin{proposition}\em
	\lab{pro1}
	Consider the model of the PEMFC+boost converter system \eqref{sys2}. Assume the state $x$ is {\em measurable} and all the system parameters are {\em known}. Fix a desired, constant value for $x_3$ as $x_3^\star >0$, and compute from \eqref{eq} the associated equilibrium vector $x^\star$. Assume $x_2^\star$ is non-negative. Consider the PI-PBC  \cite[Proposition 2]{HERetal}\footnote{See also \cite[Proposition 4.7]{ORTpidbook}.}
	\begsubequ
	\lab{pi}
	\begali{
		\lab{dotxc}
		{dx_c \over dt} & =y_N(x)\\
		\lab{u}
		u &= - K_P y_N(x) - K_I x_c,
	}
	\endsubequ
	where the input signal to the PI is defined via the function
	\begequ
	\lab{yn}
	y_N(x)= x^\star _2 x_3 -x^\star _3 x_2.
	\endequ
	For all {\em arbitrary} positive constants $K_P$ and $K_I$ we have that all signals remain bounded and
	\begequ
	\lab{limtilx3}
	\liminf {\begin{bmatrix}x(t)\\x_c(t)\end{bmatrix}}  =\begin{bmatrix}x^\star\\x_c^\star\end{bmatrix},
	\endequ
	where the equilibrium point for the controller state is given as
	\begequ
	\lab{xcsta}
	x_c^\star=-K_I^{-1}u^\star,
	\endequ 
	with $u^\star$ the univocally defined constant control associated to the equilibrium  $x^\star$.\footnote{As it is shown in \cite[Proposition B.1]{ORTpidbook} we have that $u^\star=-[g^\top(x^\star) g(x^\star)]^{-1}g^\top(x^\star)f(x^\star)$.}
\end{proposition}
\begin{proof}
	The proof follows modifying the proof of  \cite[Proposition 2]{HERetal}---to include the presence of the term $d(x_1)$---therefore, is only sketched here. 
	
	From \eqref{sys2} we obtain the error dynamics
	\begin{equation}
		\begin{aligned}\label{xtilde}
			\mathcal{Q}	\dot {\tilde x}  & = [\calj_0+\calj_1(\tilde u+u^\star) -\calr](\tilde x +x^\star)+d(x_1) \pm d(x_1^\star)\\
			& = (\calj_0+\calj_1 u^\star-\calr) x^\star + d(x_1^\star)+(\calj_0+\calj_1 u-\calr)\tilde x +\tilde u \calj_1 x^\star  +d(x_1)- d(x_1^\star)\\
			& =(\calj_0+\calj_1 u-\calr)\tilde x +\tilde u \calj_1 x^\star  +d(x_1)- d(x_1^\star)\\
		\end{aligned}
	\end{equation}
	where we used the equilibrium equation
	$$
	(\calj_0+\calj_1 u^\star-\calr) x^\star + d(x_1^\star)=0,
	$$
	to get the third identity. Now, we notice from \eqref{yn} that the (passive) output $y_N(x)$ may be written as
	\begequ
	\lab{ynx}
	y_N(x)=x^\top \calj_1 x^\star 
	\endequ
	and moreover that $y_N(x_\star)=0$, hence
	$$
	y_N(x)=y_N(\tilde x)=\tilde x^\top \calj_1x^\star.
	$$
	
	Consider the Lyapunov function candidate
	$$
	V(\tilde x,\tilde x_c):=\hal \tilde x^\top \mathcal{Q}\tilde x+{K_I \over 2}  \tilde x_c^2.
	$$
	Its time derivative satisfies
	\begalis{
		\dot V &=- \tilde x^\top \calr \tilde x + \tilde u y_N(\tilde x)  + K_I \tilde x_c y_N(\tilde x)+(x_1-x_1^\star)[d_1(x_1)- d_1(x_1^\star)]\\
		&=- \tilde x^\top \calr \tilde x + \tilde u y_N(\tilde x)  + K_I \tilde x_c y_N(\tilde x)+[V_{fc}(i_{fc})-V_{fc}(i^\star_{fc})](i_{fc} - i_{fc}^\star)\\
		&\leq - \tilde x^\top \calr \tilde x + \tilde u y_N(\tilde x)  + K_I \tilde x_c y_N(\tilde x)\\
		&=- \tilde x^\top \calr \tilde x + (u-K_Ix_c^\star) y_N(\tilde x)  + K_I \tilde x_c y_N(\tilde x)\\
		&=- \tilde x^\top \calr \tilde x + \left[ - K_P y_N(x) - K_I x_c-K_Ix_c^\star\right] y_N(\tilde x)  + K_I \tilde x_c y_N(\tilde x)\\
		&= -\tilde x^\top  \calr  \tilde x - K_P y^2_N(\tilde x)\\
		& \leq -\tilde x^\top  \calr  \tilde x\\
		&=-\theta_1 \tilde x_2^2-\theta_2 \tilde x_3^2,
	}
	where, using \eqref{Vfc} and \eqref{ifc}, in the second identity we replaced
	\begalis{
		(x_1-x_1^\star)[d_1(x_1)- d_1(x_1^\star)] &=(v_{fc}-v_{fc}^\star)(i_{fc} -i_{fc}^\star)\\
		&=[V_{fc}(i_{fc})-V_{fc}(i^\star_{fc})](i_{fc} - i_{fc}^\star),
	}
	that, in view of {\bf Lemma \ref{lem1}}, is a non positive term, yielding the first inequality and, for the third identity we invoked \eqref{xcsta}. From the inequality above and LaSalle-Yoshizawa's \cite[Theorem B.1]{ORTpidbook} we conclude that 
	\begin{align}
		\lim_{t\to\infty}\begin{bmatrix}
			x_2(t)\\x_3(t)
		\end{bmatrix}=\begin{bmatrix}
			x_{2}^\star\\x_3^\star
		\end{bmatrix}.
	\end{align}
	
	To complete the convergence proof we will show now that the vector $\col(x_2,x_3)$ is a {\em detectable} signal for the system dynamics. That is, that the following implication is true
	$$
	\begin{bmatrix}
		x_2(t)\\x_3(t)
	\end{bmatrix} \equiv \begin{bmatrix}
		x_{2}^\star\\x_3^\star
	\end{bmatrix}\quad \Rightarrow\quad \lim_{t\to\infty} x_1(t)=x_{1}^\star.
	$$
	Afterwards, we replace the left hand side equality in \eqref{sys2} to get the equations
	\begin{subequations}
		\begin{align}
			C_{fc}\dot \xi=&d_1(\xi)-x_2^\star\label{equ1}\\
			0=&\xi-\theta_1x_2^\star - x_3^\star u\label{equ2}\\
			0=&-\theta_2 x_3^\star + x_2^\star u,\label{equ3}
		\end{align}
	\end{subequations}
	where, to avoid confusion, we used the symbol $\xi$ to denote the behavior of $x_1$ in the projected dynamics. Multiplying \eqref{equ2}  by $x_2^\star$ and \eqref{equ3} by $x_3^\star$ and then summing both equations, yields
	\begin{align}\label{equ4}
		0=\xi x_2^\star -\theta_1(x_2^\star)^2 - \theta_2(x_3^\star)^2. 
	\end{align}
	From which we conclude that $\xi$ is constant. Therefore, from \eqref{equ1}, we get $x_2^\star=d_1(\xi)$. Substituting the later into \eqref{equ4} yields the equilibrium equation $p(\xi,x_3^\star,\theta)=0$ of \eqref{eq}.  Thus,
	$$
	\lim_{t\to\infty}x_1(t)=x_1^\star.
	$$
\end{proof}

\begin{remark}\em
\lab{rem2}
The proof of stability of the PI-PBC given in \cite[Proposition B.1]{ORTpidbook} and \cite[Proposition 2]{HERetal} is established showing that the signal $y_N(x)$ defined in \eqref{yn} is a {\em passive output}, hence the qualifier PBC. This together with the fact that a PI controller defines a strictly passive operator  \cite[Lemma 2.1]{ORTpidbook} for all positive PI tuning gains proves, via the Passivity Theorem and some signal chasing, that the zero equilibrium of the closed-loop system is globally asymptotically stable. 
\end{remark}

\begin{remark}\em
\lab{rem3}
Some robustness properties of the proposed PI-PBC---particularly for the case when the equilibrium point is not exactly known---may be found in \cite{ZONetalijrnlc22}.
\end{remark}

\section{Estimation of $R_p$ and $R_L$ and Adaptive PI-PBC}
\lab{sec4}
%
In this section we present the adaptive version of the PI-PBC of Proposition \ref{pro1} and prove that this version is also stabilizing---alas, only {\em practically}. 

\subsection{Proposed estimator}
\lab{subsec41}
\begin{proposition}\em
	\lab{pro2}
	Consider the model of the PEMFC+boost converter system \eqref{sys2}. Assume the state $x$ is {\em measurable} and define the {\em unknown} parameter vector 
	$\theta$. Define the I\&I parameter estimator \cite{ASTKARORTbook}
	\begin{subequations}\label{estim}
		\begin{align}
			\dot z_1&=k_1x_2\left(x_1 - z_1x_2 +{k_1 \over 2}  L x^3_2  - x_3 u\right)\label{estim1}\\
			\dot z_2&=k_2x_3\left(x_2 u -  z_2x_3 +{k_2 \over 2} C x^3_3\right)\label{estim2}\\
			\hat\theta_1&=z_1 -{k_1 \over 2} Lx^2_2\label{estim3}\\
			\hat\theta_2&=z_2 -{k_2 \over 2} Cx^2_3,\label{estim4}
		\end{align}
	\end{subequations}
	with $k_1>0$ and $k_2>0$. Assume $x_2(t) \notin \call_2$ and $x_3(t) \notin \call_2$. Then, the parameter estimation error vector $\tilde \theta$ satisfies
	$$
	\liminf |\tilde \theta(t)|=0,
	$$		
	with all signals remaining bounded
\end{proposition}
\begin{proof}
	The time derivative of $\tilde {\theta}_1=\hat{\theta}_1 - \theta_1$ is
	\begin{align*}
		\dot{\tilde{\theta}}_1=& \dot z_1 - k_1Lx_2\dot x_2\\
		=&k_1x_2(x_1 - z_1x_2 +{k_1 \over 2}  L x^3_2  - x_3 u)- k_1x_2(-(\hat\theta_1-\tilde\theta_2)x_2+x_1-x_3u)\\
		=&-k_1x_2(x_3 u  +\hat{\theta}_1x_2 )- k_1x_2(-(\hat\theta_1-\tilde\theta_1)x_2+x_1-x_3u)\\
		=&-k_1 x_2^2\tilde{\theta}_1,
	\end{align*}
	where \eqref{estim3} was used to obtain the third equality. With a similar procedure for $\tilde {\theta}_2=\hat{\theta}_2 - \theta_2$ we get, 
	\begin{align*}
		\dot{\tilde{\theta}}_2=& \dot z_2 - k_2Cx_3\dot x_3\\
		=&k_2x_3(x_2 u -  z_2x_3 +{k_2 \over 2} C x^3_3)- k_2x_3(-(\hat\theta_2-\tilde\theta_2)x_3+x_2u)\\
		=&k_2x_3(x_2 u  -\hat{\theta}_2x_3 )- k_2x_3(-(\hat\theta_2-\tilde\theta_2)x_3+x_2u)\\
		=&-k_2 x_3^2\tilde{\theta}_2.
	\end{align*}
	It follows from the last two error dynamics that $\tilde{\theta}\to 0$ iff $x_i\not\in\mathcal{L}_2,\;i=1,2$. 
\end{proof}
\subsection{Adaptive PI-PBC}

\lab{subsec42}
In this subsection we propose an adaptive PI PBC, which ensures {\em practical stability} of the desired equilibrium point---that is, the error trajectories eventually enter a residual neighborhood of the origin \cite{LAKLEEMARbook}. 

\begin{proposition}\em\label{pro3}
	Fix a desired, constant value for $x_3$ as $x_3^\star >0$. Consider \eqref{sys2} in closed-loop with the adaptive PI-PBC controller 
	\begsubequ
	\lab{adapi}
	\begali{
		\lab{adadotxc}
		\dot x_c  & =\hat y_N(x)\\
		\lab{adau}
		u &= - K_P \hat y_N(x) - K_I x_c,
	}
	\endsubequ
	where the input signal to the PI is defined as
	\begequ
	\lab{hatyn}
	\hat y_N( x)= \hat x^\star _2 x_3 -x^\star _3 x_2.
	\endequ
	with $\hat x_2^\star=d_1(\hat x^\star_1)$ where $\hat x^\star_1$ is the solution of the estimated equilibrium equation
	\begalis{
		d_1(\hat x^\star_1)\hat x^\star_1-\hat {\theta}_2 (x^\star_{3})^2-\hat {\theta}_1 d^2_1(\hat  x_1^\star)&=0,
	}
	and the estimates $\hat \theta_1,\hat \theta_2$ are obtained with I\&I estimator \eqref{estim}. Assume $x_2(t) \notin \call_2$ and $x_3(t) \notin \call_2$. Then, the zero equilibrium of the closed-loop error system is {\em practically stable} and the size of the residual set can be reduced increasing $K_P$.
\end{proposition}

\begin{proof}
	First, notice that the input to the adaptive PI-PBC given in \eqref{hatyn} may be written as 
	\begali{
		\nonumber
		\hat y_N( x) =&x^\top\mathcal{J}_1 \hat x^\star \\
		\nonumber	=&x^\top\mathcal{J}_1(e+x^\star) \\
		\lab{hatynyn}	=&y_N(x)+x^\top\mathcal{J}_1e
	}
	where we used \eqref{e} in the first identity and \eqref{yn} in the second one. Second, let us denote the equilibrium value of $x_c$ as $\bar x_c \in \rea$ and the equilibrium error as
	$$
	x_c^e:=x_c - \bar x_c.
	$$
	Denote the equilibrium value for $u$ as $\bar u \in \rea$, from \eqref{adau}, we see that
	$$
	\bar u = - K_I \bar x_c,
	$$
	and denote the input error signal as
	$$
	u_e:=u-\bar u.
	$$   
	
	From the proof of Proposition \ref{pro1} we have that the time derivative of the function $V_x(\tilde x):=\hal \tilde x^\top \mathcal{Q}\tilde x$, along the trajectories of \eqref{xtilde}, satisfies the bound
	\begalis{
		\dot V_x &\leq - \tilde x^\top \calr \tilde x + \tilde u y_N(x).
	} 
	On the other hand, the time derivative of the function $V_{x_c}(x^e_c):={K_I \over 2} (x^e_c)^2$ along the trajectories of \eqref{adadotxc} yields 
	\begalis{
		\dot V_{x^e_c}&= K_I x^e_c \dot x_c \\
		&= K_I x_c \hat y_N(x) - K_I \bar x_c \hat y_N(x)\\
		&=  - [u+K_P \hat y_N(x)]\hat y_N(x)  - K_I \bar x_c \hat y_N(x)\\
		&=  - K_P \hat y^2_N(x) -(u + K_I \bar x_c)  \hat y_N(x)\\
		&=  - K_P \hat y^2_N(x) -u_e  \hat y_N(x). 
	}

	Combining these two equations we get
	\begequ
	\lab{dotv}
	\dot V_x + \dot V_{x_c} \leq - \tilde x^\top \calr \tilde x- K_P  \hat y_N^2(x) +\xi, 
	\endequ
	where we defined the signal 
	\begequ
	\lab{xi}
	\xi:= -u_e  \hat y_N(x) + \tilde u y_N(x).
	\endequ
	Consider now the following chain of implications
	\begin{align}
		x_2(t)  \notin \call_2\;\mbox{and}\; x_3(t) \notin \call_2 & \quad \dif \quad |\tilde \theta(t)|\;\to\;0 \\
		& \quad \ttt \quad |e(t)|\;\to\;0 \\
		& \quad \fff \quad |\hat y_N(t) - y_N(t)|\;\to\;0 \\
		& \quad \hhh \quad |\xi(t)-(\bar u - u^\star)\hat y_N(x(t))|\;\to\;0,
	\end{align}
	where the last implication stems from the identity
	$$
	\tilde u-u_e= \bar u - u^\star.  
	$$
	The proof is completed replacing the last limit in the third right hand term of the inequality \eqref{dotv} and noting that the right hand term of this equation contains a term $- K_P \hat y^2_N(x)$ that, for large $x$, dominates the linear term $(\bar u - u^\star)\hat y_N(x)$. 
\end{proof}

\section{Simulation Results}
\lab{sec5}
%
The system is described in detail in \cite{BELetal}---the reader is referred to it for further details on the components and their characteristics. In Fig. \ref{vfc_ifc}, an experimental curve of its PEMFC is shown. From the model \eqref{Vfc}, an estimation of that curve was obtained and its graph is depicted in the same figure. The parameters $c_i$ of the estimated curve were obtained using the Least Squares method.

We simulate three different scenarios. In all of them, we consider that the PEMFC behaves according to the estimated curve. The numerical values of the parameters $c_1$ to $c_5$, as well as the other parameter of the system used in the simulations, are given in Table \ref{table1}. The first simulation consists in testing the regulation performance when $R_L$ and $R_p$ are both known, therefore, the equilibrium points are also available from their computation. Later,  we assess robustness of the controller when there is a load variation and no update of the new load is available. Finally, in the last scenario, we simulate the adaptive controller, that is, we estimate $R_L$ and $R_p$ and compute the equilibrium points from those values. The load is varied during the simulation.

\begin{table}[b!]
	\centering
	\begin{tabular}[t!]{ cc } 
		\toprule
		\midrule
		Parameter & Nominal value\\
		\midrule
		$C_{fc}$ & $50{\mathrm mF}$ \\ 
		$C$ &  $1.5{\mathrm mF}$ \\ 
		$L$ &   $36.1{\mathrm \mu H}$\\
		$R_L$ & $4.608\Omega$ \\
		$R_p$ & $0.1\Omega$ \\
		$K_I$& 0.001\\
		$K_P$& 1\\
		$c_1$& 39.3543\\
		$c_2$ & 2.5825 \\
		$c_3$ & 0.1808\\
		$c_4$ & 0.0046\\
		$c_5$ & 1.2610\\
		\midrule
		\toprule
	\end{tabular}
	\caption{System parameters}
	\label{table1}
\end{table}
\subsection{Regulation performance when $R_L$ and $R_p$ are known}
In this test, we numerically verified the PI-PBC of Proposition \ref{pro1}. It is assumed that $R_L$ and $R_p$ are known and $x^\star$ is computed from these values. To avoid large overshoots, the system's initial conditions were set as follows: $(x(0),x_c(0))=(40\mathrm{V}, 10\mathrm{A}, 30\mathrm{V},0)$. Since the values of duty cycle are restricted to the interval  $(0,1)$, the control signal $u$ was limited by a saturation block before entering into the system.  

During the simulation, we varied the voltage setpoint. Initially, we set $x_{3}^\star=45\mathrm{V}$. At time $t=0.25\mathrm{s}$,  that setpoint takes the new value of $x_{3}^\star=60\mathrm{V}$---at this time, the equilibrium values are updated in the simulation according to this new setpoint and their values are shown in the figure. As it is seen from Fig. \ref{reg}, the voltage regulation error converges to zero in both cases. The same can be observed about the error of the two other variables with respect to their equilibrium value. On the other hand, the control signal $u$ exhibits saturation when the simulation started, but it enters its valid operation interval quickly afterwards.

\begin{figure}[t!]
	\centering
	\includegraphics[scale=0.60]{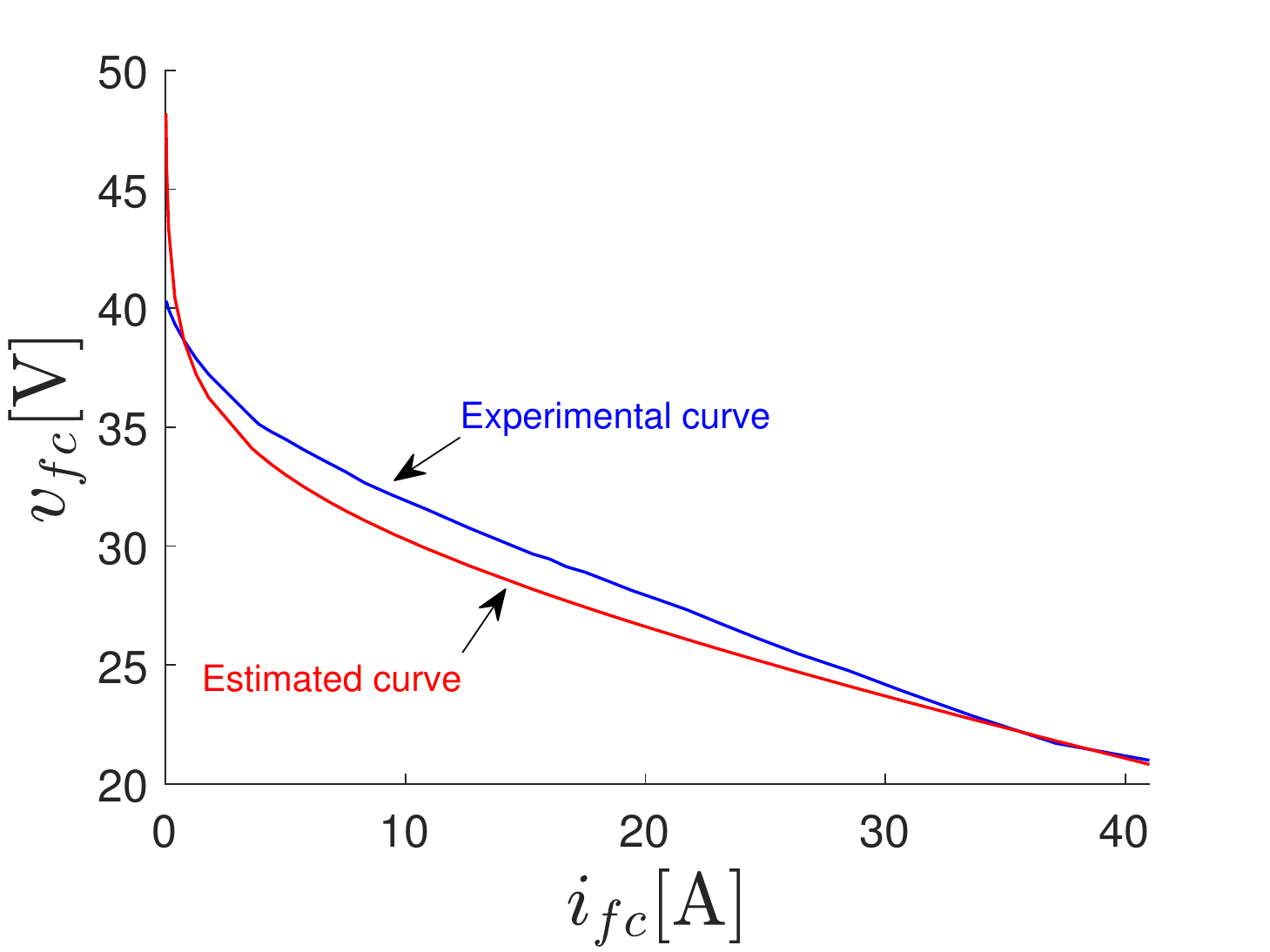}
	\caption{Experimental curve and its estimated curve \eqref{Vfc}---see Table \ref{table1} for the numerical values of $c_i$.}
	\label{vfc_ifc}
\end{figure}
\begin{figure*}
	\centering
	\includegraphics[scale=0.60]{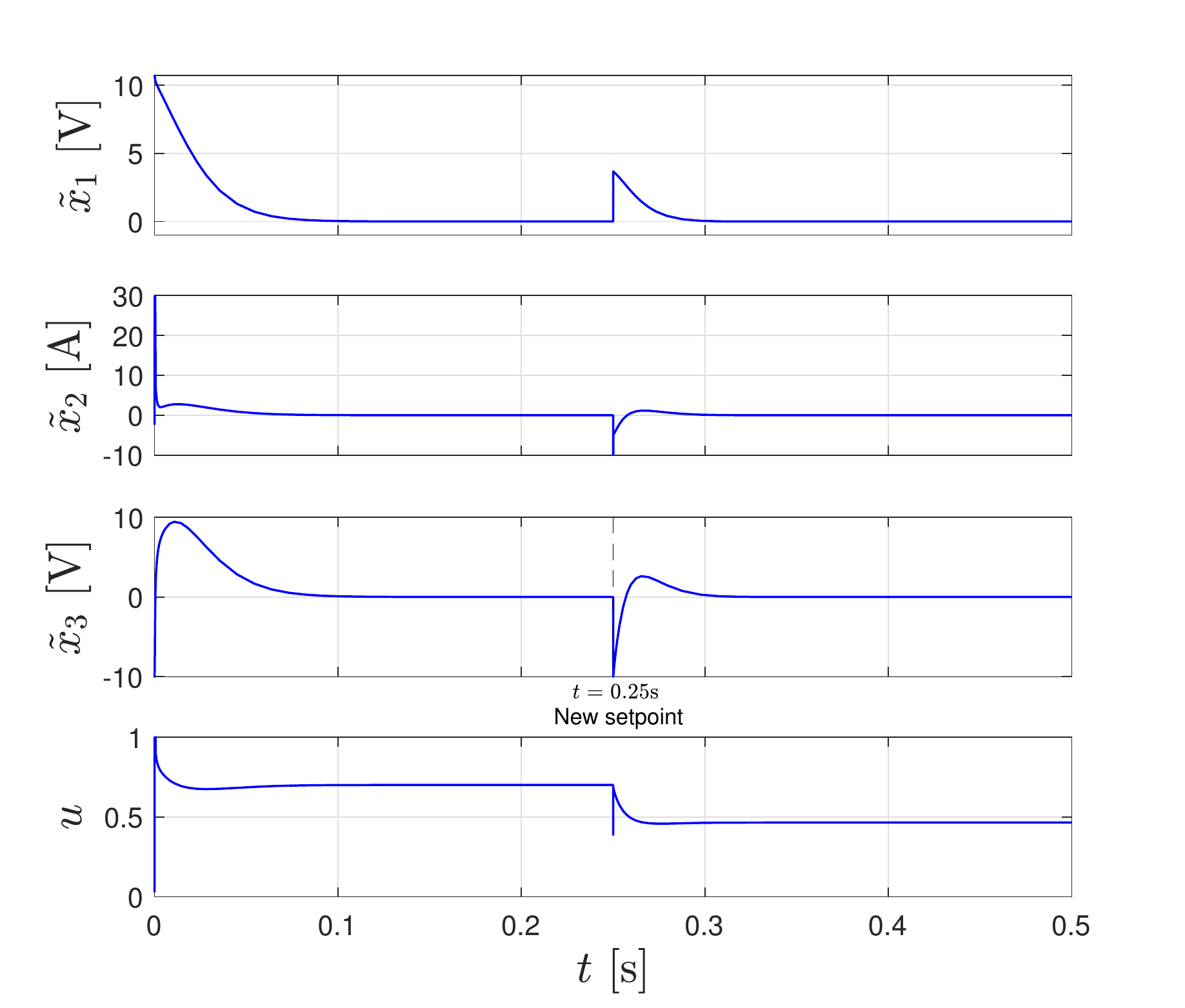}
	\caption{State regulation errors and control signal when all parameters are known. Equilibrium values: a)for $t<0.25$s, $[x^\star]_1=(29.28\mathrm{V},12.38\mathrm{A},40\mathrm{V})$, b)for $t\geq 0.25$s, $[x^\star]_2=(25.6\mathrm{V},23.31\mathrm{A},50\mathrm{V})$. }   
	\label{reg}
\end{figure*}

\subsection{Robustness to load changes}
In this test, we assess robustness to load variations of the PI-PBC. We assumed that $R_p$ and $R_L$ are known when simulation starts. From the beginning, the system is operating in steady state with the output voltage $x_3$ regulated at a desired value of $40$V. At $t=0.2$s, a change in $R_L$ is produced. The load varies from its nominal value of $4.608\Omega$ to $3.9168\mathrm{\Omega}$. The new value of $R_L$ is not updated in the controller, yielding an erroneous  value of $x^\star$. This results in a steady-state regulation error. Figure \ref{rob} shows how the voltages and current evolve during the simulation. It can be observed from the plot that, after the load variation, the steady-state output voltage decays to $35$V, that is, $5$V below its reference. We remark, however, that the voltages $x_1$ and $x_3$, and the current $x_2$ remained bounded.

\begin{figure}[t!]
	\centering
	\includegraphics[scale=0.60]{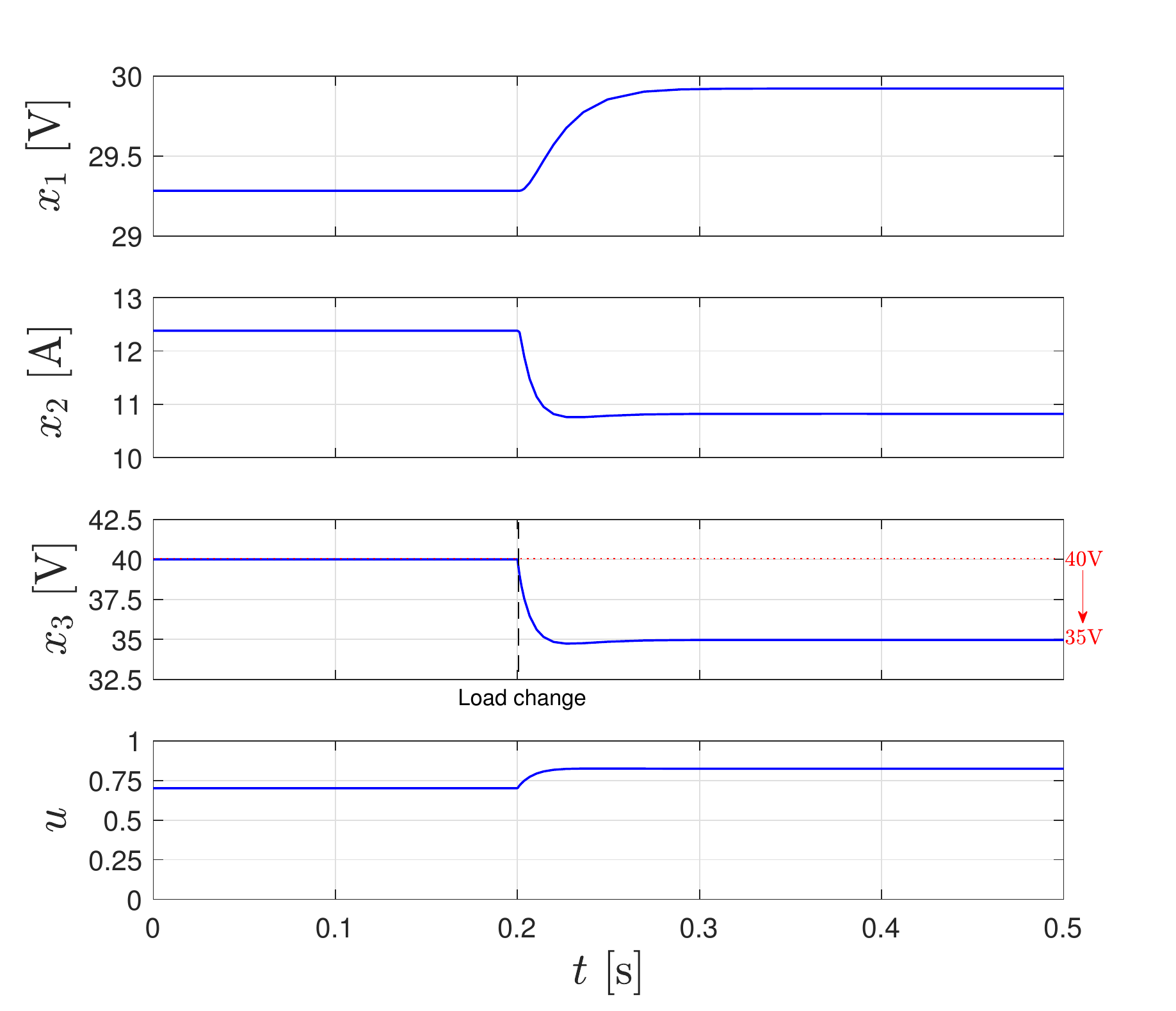}
	\caption{Regulation performance of the PI-PBC when $R_L$ changes.}
	\label{rob}
\end{figure}

\subsection{Regulation performance of the adaptive controller}
In this last test, $R_L$ and $R_p$ are assumed to be unknown. We simulated the proposed I\&I estimator operating simultaneously with the PI-PBC---see Proposition \ref{pro3}. Since the PI-PBC requires the equilibrium values, estimations of them were calculated from estimates of the resistances provided by the estimator. That is, from \eqref{eq}, for each $\hat\theta$, we compute
$$
p(x_1,x_3^\star,\hat\theta):=d_1(x_1)x_1 - \hat{\theta}_2x_{3\star}^2-\hat{\theta}_1d_1^2(x_1),
$$
for different values of $x_1$ within the operation range $x_1\in[21,48]$ and we arrange those values in a vector. Using this vector, an approximate of the root of $p(\cdot)$, i.e., the estimate of $x_1^\star$, was online computed from
$$
\hat{x}_1^\star=\underset{x_1}{\argmin}\;|p(x_1,x_{3\star},\hat\theta)|.
$$
Notice that computation of $p(x_1,x_3^\star,\hat\theta)$ requires calculating $d(x_1)$ in advance. Due to the difficulty of finding an analytical expression of $d_1(x_1)$, a lookup table was employed that associates for each $x_1$ its corresponding value $d_1(x_1)$.  The estimate of $x_2^\star$ was obtained from the equality $\hat{x}_2^\star=d_1(\hat{x}_1^\star)$.

We initialized the system at $(x(0),x_c(0),\hat \theta)=(40\mathrm{V}, 10\mathrm{A}, 30\mathrm{V},0,0,0)$. The output voltage setpoint was set at $x_3^\star=40$V and maintained to that level during the simulation. The initial value of $R_L$ was that shown in the table. At $t=0.25$s, $R_L$ was changed to $85\%$ of its nominal value.  Fig. \ref{adaptive1} depicts how $\tilde\theta$ evolves before/after the load change for different gain values $k_i$. As it can be expected from the error dynamics,  the greater the values, the faster the convergence. Notice that, when $k_1=k_2=0.01$, it is not possible for the estimator to converge during the simulation time. On the other hand, when $k_1=k_2=10$ the convergence is almost instantaneous. 

In Fig. \ref{adaptive2}, we show the results of the adaptive PI-PBC when the estimator gains are set to $k_1=k_2=10$. As depicted, the system is regulated to the desired output voltage once $R_L$ is estimated. As mentioned, at $t=0.25$s, the load varies from  its nominal value to $R_L=3.9168\mathrm{\Omega}$. A transient behavior is produced due to this variation, however, the output voltage converges back to the imposed $x_3^\star$---the error $\tilde x_3$ settles down to zero. This response contrasts with that of the previous test in which the controller was incapable of maintaining the voltage at the desired value after this load variation.

\begin{figure}[h!]
	\centering
	\begin{subfigure}[Convergence of $\tilde{\theta}$ for different estimator gains $k_i$.]
		{\includegraphics[scale=0.62]{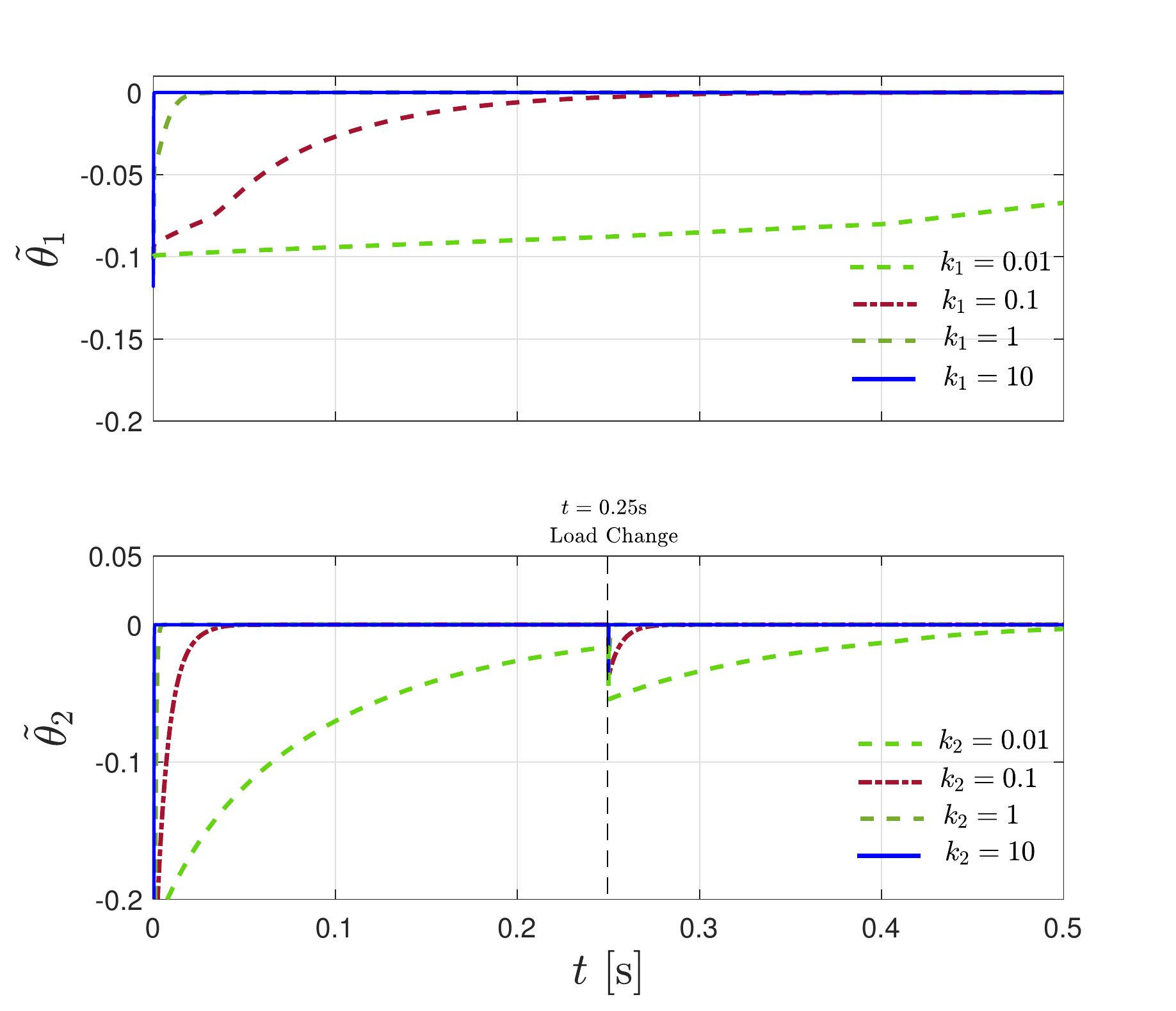}\label{adaptive1}}
	\end{subfigure}
	\begin{subfigure}[$\tilde x$ and $u$ when $k_1=k_2=10$.] 
		{\includegraphics[scale=0.62]{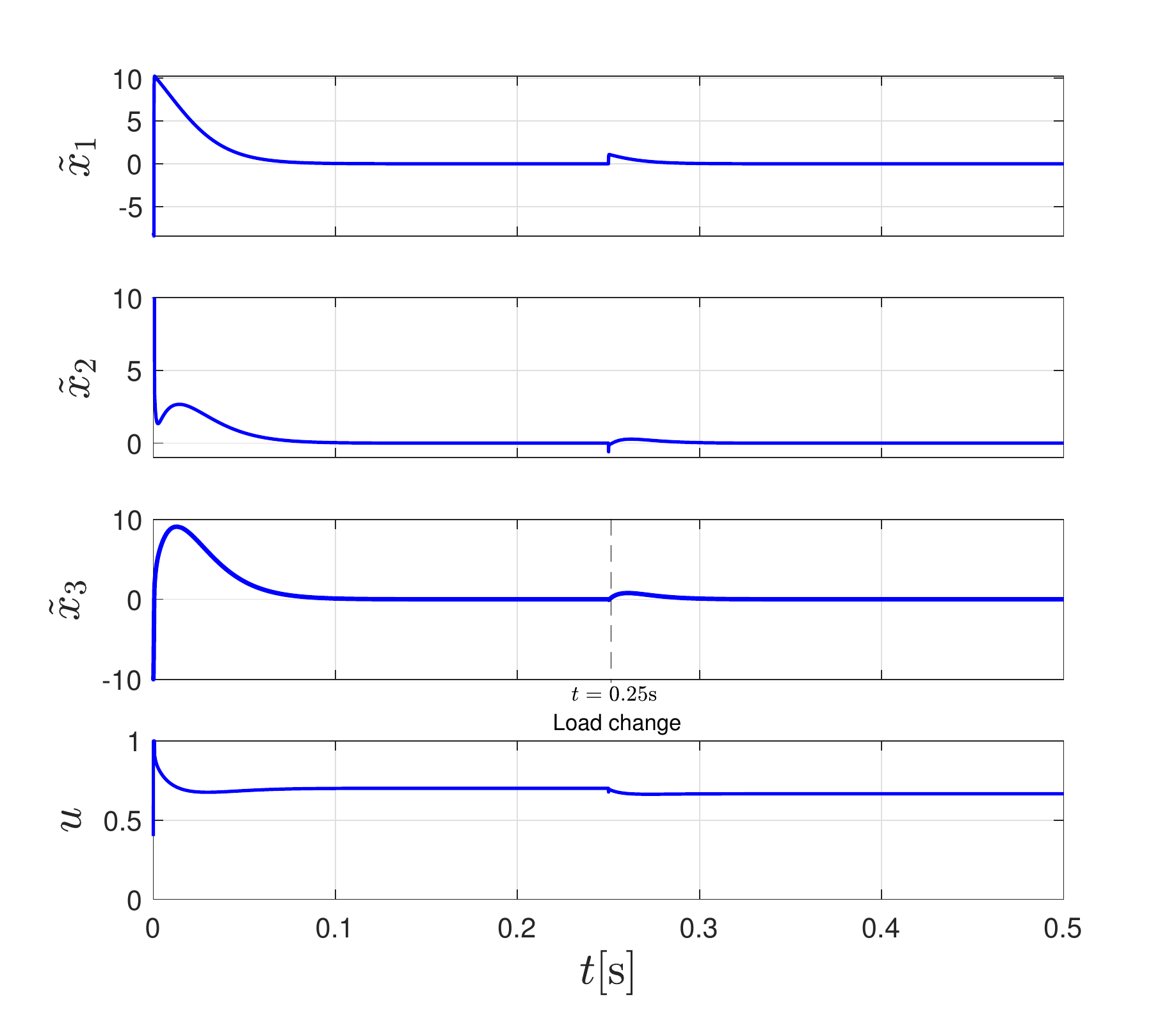}\label{adaptive2}}
	\end{subfigure}
	\caption{Regulation performance of the adaptive PI-PBC when $R_L$ changes.}\label{adaptive}
\end{figure}

\section{Concluding Remarks}
\lab{sec6}
%

In this contribution, we propose an adaptive controller that regulates the output voltage of a system composed of a fuel cell, boost converter and a resistive load. The behavior of the fuel cell is modeled by a static mapping relating the cell's current with its voltage. Benefiting from the monotonicity of this mapping, we show that---in spite of the presence of the cell's current---the passive output of the system is still preserved, thus, it can be regulated by a PI-PBC. Furthermore, an adaptive scheme is proposed whose operation combined with that of the PI-PBC allows the estimation of the equilibrium points, required by the controller, and the regulation of the output voltage at a desired level. Stability of the overall system is proved and numerical results are presented. Ongoing experimental implementation is intended for physical validation of the proposed adaptive PI-PBC and it will be reported in the near future.\\


\bibliographystyle{ieeetr}
\bibliography{refs}

\end{document}